\documentclass[conference, twocolumn]{IEEEtran}
\usepackage{amssymb, amsmath}
\usepackage{amsfonts}
\usepackage{bbm}
\usepackage{mathrsfs}
\usepackage{xspace}
\usepackage{bm}

\usepackage{cite}
\usepackage{epsfig}
\usepackage{cases}
\usepackage{psfrag}
\usepackage{enumerate}

\usepackage{makeidx}
\usepackage{supertabular}
\usepackage{acronym}

\newtheorem{theorem}{Theorem}

\newtheorem{proposition}{Proposition}

\newtheorem{notation}{Notation}

\newenvironment{textbmatrix}{   \setlength{\arraycolsep}{2.5pt}%
                                                                \big[\begin{matrix}}{\end{matrix}\big]%
                                                                \raisebox{0.08ex}{\vphantom{M}}}

\def\be{\begin{equation}}
\def\ee{\end{equation}}
\def\een{\nonumber \end{equation}}
\def\mat{\begin{bmatrix}}
\def\emat{\end{bmatrix}}
\def\btm{\begin{textbmatrix}}
\def\etm{\end{textbmatrix}}

\def\ba#1\ea{\begin{align}#1\end{align}}
\def\bs#1\es{\begin{split}#1\end{split}}
\def\bg#1\eg{\begin{gather}#1\end{gather}}
\def\bi#1\ei{\begin{itemize}#1\end{itemize}}

\newcommand{\safemath}[2]{\newcommand{#1}{\ensuremath{#2}\xspace}}

\newcommand{\lefto}{\mathopen{}\left}

\DeclareMathOperator{\Varop}{\mathbb{V}\!\mathrm{ar}} % variance operator
\safemath{\interior}{\mathrm{Int}}                       % interior of a set
\newcommand{\inv}[1]{\ensuremath{#1^{-1}}}      % matrix inverse
\safemath{\dfn}{:=}                                                     % definition
\safemath{\dirac}{\delta}                                       % Dirac delta

\safemath{\No}{N_0}                                                     % noise spectral density
\safemath{\Es}{E_s}                                                     % energy per symbol
\safemath{\Eb}{E_b}                                                     % energy per bit
\safemath{\EbNo}{\frac{\Eb}{\No}}
\safemath{\EsNo}{\frac{\Es}{\No}}

\DeclareMathOperator{\CHop}{\ensuremath{\mathbb{H}}} % channel operator
\safemath{\tvir}{h_{\CHop}}                                     % time-varying impulse response
\safemath{\tvtf}{L_{\CHop}}                                     %       transfer function
\safemath{\spf}{S_{\CHop}}
\safemath{\bff}{H_{\CHop}}                                      %       function

\safemath{\ircf}{R_{h}}                                         % impulse response correlation fn.
\safemath{\scf}{R_{S}}                                          % scattering function
\safemath{\tfcf}{R_{L}}                                         % time-frequency correlation fn.
\safemath{\bfcf}{R_{H}}                                         % bi-frequency correlation fn.
\safemath{\mi}{I}                                                       % muttal information
\safemath{\capacity}{C}                                         % capacity

\newcommand{\pdf[2]}{f_{{#1}}\left({#2}\right)}                    % probability density function
\safemath{\uniform}{\mathcal{U}}                        % uniform distribution
\safemath{\normal}{\mathcal{N}}                         % normal distribution
\safemath{\circnorm}{\mathcal{CN}}                      % circ. symm. normal
\safemath{\mchain}{\leftrightarrow}                     % Markov chain

\safemath{\dB}{\,\mathrm{dB}}
\safemath{\dBm}{\,\mathrm{dBm}}
\safemath{\Hz}{\,\mathrm{Hz}}
\safemath{\kHz}{\,\mathrm{kHz}}
\safemath{\MHz}{\,\mathrm{MHz}}
\safemath{\GHz}{\,\mathrm{GHz}}
\safemath{\s}{\,\mathrm{s}}
\safemath{\ms}{\,\mathrm{ms}}
\safemath{\mus}{\,\mathrm{\mu s}}
\safemath{\ns}{\,\mathrm{ns}}
\safemath{\meter}{\,\mathrm{m}}
\safemath{\mm}{\,\mathrm{mm}}
\safemath{\cm}{\,\mathrm{cm}}
\safemath{\m}{\,\mathrm{m}}
\safemath{\W}{\,\mathrm{W}}
\safemath{\J}{\,\mathrm{J}}
\safemath{\K}{\,\mathrm{K}}
\safemath{\bit}{\,\mathrm{bit}}
\safemath{\mW}{\,\mathrm{mW}}
\safemath{\nW}{\,\mathrm{nW}}
\safemath{\pW}{\,\mathrm{pW}}
\safemath{\muW}{\,\mu\mathrm{W}}
\safemath{\Watt}{\,\mathrm{W}}
\safemath{\kbps}{\,\mathrm{kb/s}}
\safemath{\Mbps}{\,\mathrm{Mb/s}}
\safemath{\bpsHz}{\,\mathrm{b/s/Hz}}

\safemath{\define}{\triangleq}                  % definition

\safemath{\equivalent}{\sim}
\safemath{\distas}{\sim}                                        % distributed according to

\safemath{\reals}{\mathbb{R}}
\safemath{\positivereals}{\mathbb{R}^{+}}
\safemath{\integers}{\mathbb{Z}}
\safemath{\posint}{\mathbb{Z}_{+}}
\safemath{\naturals}{\mathbb{N}}
\safemath{\complexset}{\mathbb{C}}

\safemath{\setA}{\mathcal{A}}
\safemath{\setB}{\mathcal{B}}
\safemath{\setC}{\mathcal{C}}
\safemath{\setD}{\mathcal{D}}
\safemath{\setE}{\mathcal{E}}
\safemath{\setF}{\mathcal{F}}
\safemath{\setG}{\mathcal{G}}
\safemath{\setH}{\mathcal{H}}
\safemath{\setI}{\mathcal{I}}
\safemath{\setJ}{\mathcal{J}}
\safemath{\setK}{\mathcal{K}}
\safemath{\setL}{\mathcal{L}}
\safemath{\setM}{\mathcal{M}}
\safemath{\setN}{\mathcal{N}}
\safemath{\setO}{\mathcal{O}}
\safemath{\setP}{\mathcal{P}}
\safemath{\setQ}{\mathcal{Q}}
\safemath{\setR}{\mathcal{R}}
\safemath{\setS}{\mathcal{S}}
\safemath{\setT}{\mathcal{T}}
\safemath{\setU}{\mathcal{U}}
\safemath{\setV}{\mathcal{V}}
\safemath{\setW}{\mathcal{W}}
\safemath{\setX}{\mathcal{X}}
\safemath{\setY}{\mathcal{Y}}
\safemath{\setZ}{\mathcal{Z}}
\safemath{\emptySet}{\varnothing}

\safemath{\bma}{\mathbf{a}}
\safemath{\bmb}{\mathbf{b}}
\safemath{\bmc}{\mathbf{c}}
\safemath{\bmd}{\mathbf{d}}
\safemath{\bme}{\mathbf{e}}
\safemath{\bmf}{\mathbf{f}}
\safemath{\bmg}{\mathbf{g}}
\safemath{\bmh}{\mathbf{h}}
\safemath{\bmi}{\mathbf{i}}
\safemath{\bmj}{\mathbf{j}}
\safemath{\bmk}{\mathbf{k}}
\safemath{\bml}{\mathbf{l}}
\safemath{\bmm}{\mathbf{m}}
\safemath{\bmn}{\mathbf{n}}
\safemath{\bmo}{\mathbf{o}}
\safemath{\bmp}{\mathbf{p}}
\safemath{\bmq}{\mathbf{q}}
\safemath{\bmr}{\mathbf{r}}
\safemath{\bms}{\mathbf{s}}
\safemath{\bmt}{\mathbf{t}}
\safemath{\bmu}{\mathbf{u}}
\safemath{\bmv}{\mathbf{v}}
\safemath{\bmw}{\mathbf{w}}
\safemath{\bmx}{\mathbf{x}}
\safemath{\bmy}{\mathbf{y}}
\safemath{\bmz}{\mathbf{z}}

\bmdefine{\biad}{a}
\bmdefine{\bibd}{b}
\bmdefine{\bicd}{c}
\bmdefine{\bidd}{d}
\bmdefine{\bied}{e}
\bmdefine{\bifd}{f}
\bmdefine{\bigd}{g}
\bmdefine{\bihd}{h}
\bmdefine{\biid}{i}
\bmdefine{\bijd}{j}
\bmdefine{\bikd}{k}
\bmdefine{\bild}{l}
\bmdefine{\bimd}{m}
\bmdefine{\bind}{n}
\bmdefine{\biod}{o}
\bmdefine{\bipd}{p}
\bmdefine{\biqd}{q}
\bmdefine{\bird}{r}
\bmdefine{\bisd}{s}
\bmdefine{\bitd}{t}
\bmdefine{\biud}{u}
\bmdefine{\bivd}{v}
\bmdefine{\biwd}{w}
\bmdefine{\bixd}{x}
\bmdefine{\biyd}{y}
\bmdefine{\bizd}{z}

\bmdefine{\bixid}{\xi}
\bmdefine{\bilambdad}{\lambda}
\bmdefine{\bimud}{\mu}
\bmdefine{\bithetad}{\theta}
\bmdefine{\biphid}{\phi}

\safemath{\bmia}{\biad}
\safemath{\bmib}{\bibd}
\safemath{\bmic}{\bicd}
\safemath{\bmid}{\bidd}
\safemath{\bmie}{\bied}
\safemath{\bmif}{\bifd}
\safemath{\bmig}{\bigd}
\safemath{\bmih}{\bihd}
\safemath{\bmii}{\biid}
\safemath{\bmij}{\bijd}
\safemath{\bmik}{\bikd}
\safemath{\bmil}{\bild}
\safemath{\bmim}{\bimd}
\safemath{\bmin}{\bind}
\safemath{\bmio}{\biod}
\safemath{\bmip}{\bipd}
\safemath{\bmiq}{\biqd}
\safemath{\bmir}{\bird}
\safemath{\bmis}{\bisd}
\safemath{\bmit}{\bitd}
\safemath{\bmiu}{\biud}
\safemath{\bmiv}{\bivd}
\safemath{\bmiw}{\biwd}
\safemath{\bmix}{\bixd}
\safemath{\bmiy}{\biyd}
\safemath{\bmiz}{\bizd}

\safemath{\bmxi}{\bixid}
\safemath{\bmlambda}{\bilambdad}
\safemath{\bmmu}{\bimud}
\safemath{\bmtheta}{\bithetad}
\safemath{\bmphi}{\biphid}

\safemath{\bA}{\mathbf{A}}
\safemath{\bB}{\mathbf{B}}
\safemath{\bC}{\mathbf{C}}
\safemath{\bD}{\mathbf{D}}
\safemath{\bE}{\mathbf{E}}
\safemath{\bF}{\mathbf{F}}
\safemath{\bG}{\mathbf{G}}
\safemath{\bH}{\mathbf{H}}
\safemath{\bI}{\mathbf{I}}
\safemath{\bJ}{\mathbf{J}}
\safemath{\bK}{\mathbf{K}}
\safemath{\bL}{\mathbf{L}}
\safemath{\bM}{\mathbf{M}}
\safemath{\bN}{\mathbf{N}}
\safemath{\bO}{\mathbf{O}}
\safemath{\bP}{\mathbf{P}}
\safemath{\bQ}{\mathbf{Q}}
\safemath{\bR}{\mathbf{R}}
\safemath{\bS}{\mathbf{S}}
\safemath{\bT}{\mathbf{T}}
\safemath{\bU}{\mathbf{U}}
\safemath{\bV}{\mathbf{V}}
\safemath{\bW}{\mathbf{W}}
\safemath{\bX}{\mathbf{X}}
\safemath{\bY}{\mathbf{Y}}
\safemath{\bZ}{\mathbf{Z}}

\bmdefine{\biAd}{A}
\bmdefine{\biBd}{B}
\bmdefine{\biCd}{C}
\bmdefine{\biDd}{D}
\bmdefine{\biEd}{E}
\bmdefine{\biFd}{F}
\bmdefine{\biGd}{G}
\bmdefine{\biHd}{H}
\bmdefine{\biId}{I}
\bmdefine{\biJd}{J}
\bmdefine{\biKd}{K}
\bmdefine{\biLd}{L}
\bmdefine{\biMd}{M}
\bmdefine{\biOd}{N}
\bmdefine{\biPd}{O}
\bmdefine{\biQd}{P}
\bmdefine{\biRd}{R}
\bmdefine{\biSd}{S}
\bmdefine{\biTd}{T}
\bmdefine{\biUd}{U}
\bmdefine{\biVd}{V}
\bmdefine{\biWd}{W}
\bmdefine{\biXd}{X}
\bmdefine{\biYd}{Y}
\bmdefine{\biZd}{Z}

\bmdefine{\biDelta}{\Delta}
\bmdefine{\biLambda}{\Lambda}
\bmdefine{\biPhi}{\Phi}
\bmdefine{\biSigma}{\Sigma}
\bmdefine{\biOmega}{\Omega}
\bmdefine{\biTheta}{\Theta}

\safemath{\bimA}{\biAd}
\safemath{\bimB}{\biBd}
\safemath{\bimC}{\biCd}
\safemath{\bimD}{\biDd}
\safemath{\bimE}{\biEd}
\safemath{\bimF}{\biFd}
\safemath{\bimG}{\biGd}
\safemath{\bimH}{\biHd}
\safemath{\bimI}{\biId}
\safemath{\bimJ}{\biJd}
\safemath{\bimK}{\biKd}
\safemath{\bimL}{\biLd}
\safemath{\bimM}{\biMd}
\safemath{\bimN}{\biNd}
\safemath{\bimO}{\biOd}
\safemath{\bimP}{\biPd}
\safemath{\bimQ}{\biQd}
\safemath{\bimR}{\biRd}
\safemath{\bimS}{\biSd}
\safemath{\bimT}{\biTd}
\safemath{\bimU}{\biUd}
\safemath{\bimV}{\biVd}
\safemath{\bimW}{\biWd}
\safemath{\bimX}{\biXd}
\safemath{\bimY}{\biYd}
\safemath{\bimZ}{\biZd}

\safemath{\bDelta}{\bielta}
\safemath{\bLambda}{\biLambda}
\safemath{\bPhi}{\biPhi}
\safemath{\bSigma}{\biSigma}
\safemath{\bOmega}{\biOmega}
\safemath{\bTheta}{\biTheta}

\safemath{\veca}{\bma}
\safemath{\vecb}{\bmb}
\safemath{\vecc}{\bmc}
\safemath{\vecd}{\bmd}
\safemath{\vece}{\bme}
\safemath{\vecf}{\bmf}
\safemath{\vecg}{\bmg}
\safemath{\vech}{\bmh}
\safemath{\veci}{\bmi}
\safemath{\vecj}{\bmj}
\safemath{\veck}{\bmk}
\safemath{\vecl}{\bml}
\safemath{\vecm}{\bmm}
\safemath{\vecn}{\bmn}
\safemath{\veco}{\bmo}
\safemath{\vecp}{\bmp}
\safemath{\vecq}{\bmq}
\safemath{\vecr}{\bmr}
\safemath{\vecs}{\bms}
\safemath{\vect}{\bmt}
\safemath{\vecu}{\bmu}
\safemath{\vecv}{\bmv}
\safemath{\vecw}{\bmw}
\safemath{\vecx}{\bmx}
\safemath{\vecy}{\bmy}
\safemath{\vecz}{\bmz}
\safemath{\vecZero}{\bZero}

\safemath{\vecxi}{\bmxi}
\safemath{\veclambda}{\bmlambda}
\safemath{\vecmu}{\bmmu}
\safemath{\vectheta}{\bmtheta}
\safemath{\vecphi}{\bmphi}

\safemath{\matA}{\bA}
\safemath{\matB}{\bB}
\safemath{\matC}{\bC}
\safemath{\matD}{\bD}
\safemath{\matE}{\bE}
\safemath{\matF}{\bF}
\safemath{\matG}{\bG}
\safemath{\matH}{\bH}
\safemath{\matI}{\bI}
\safemath{\matJ}{\bJ}
\safemath{\matK}{\bK}
\safemath{\matL}{\bL}
\safemath{\matM}{\bM}
\safemath{\matN}{\bN}
\safemath{\matO}{\bO}
\safemath{\matP}{\bP}
\safemath{\matQ}{\bQ}
\safemath{\matR}{\bR}
\safemath{\matS}{\bS}
\safemath{\matT}{\bT}
\safemath{\matU}{\bU}
\safemath{\matV}{\bV}
\safemath{\matW}{\bW}
\safemath{\matX}{\bX}
\safemath{\matY}{\bY}
\safemath{\matZ}{\bZ}
\safemath{\matZero}{\bZero}

\safemath{\matDelta}{\bDelta}
\safemath{\matLambda}{\bLambda}
\safemath{\matPhi}{\bPhi}
\safemath{\matSigma}{\bSigma}
\safemath{\matOmega}{\bOmega}
\safemath{\matTheta}{\bTheta}

\safemath{\matIdentity}{\matI}

\usepackage{fancyhdr,graphicx,color,listings}

\newcommand{\sectionname}{Sec.}

\renewcommand{\figurename}{Fig.}

\renewcommand{\tablename}{Tab.}

\newcommand{\theoremname}{Theorem}
\newcommand{\equationname}{Eq.}
\newcommand{\equationsname}{Eqs.}

\newcommand{\covs}{\mathbf{\bSigma}}

\safemath{\sourceindex}{s}
\safemath{\othersource}{s'}
\safemath{\sourceset}{\setS}
\safemath{\sourcesubset}{\setL}
\safemath{\sourcenumber}{S}

\safemath{\destindex}{d}

\safemath{\relayindex}{r}
\safemath{\otherrelay}{r'}
\safemath{\relayset}{\setR}
\safemath{\relaysubset}{\setT}
\safemath{\relaynumber}{R}

\safemath{\setname}{\setA}

\newcommand{\relayfunction[2]}{f_{#1}^{#2}}

\safemath{\timeindex}{t}

\safemath{\powerSymbol}{p}
\newcommand{\power[1]}{\powerSymbol_{#1}}
\newcommand{\powermax[1]}{\overline{\powerSymbol}_{#1}}

\safemath{\upperboundSymbol}{\mathcal{C}}
\newcommand{\upperbound[2]}{\displaystyle\upperboundSymbol_{#1}{#2}}
\safemath{\capacitySymbol}{C}
\safemath{\SNRSymbol}{\gamma}
\newcommand{\SNR[2]}{\SNRSymbol_{{#1}{#2}}}

\safemath{\pathgainSymbol}{h}
\newcommand{\pathgain[2]}{\pathgainSymbol_{{#1}{#2}}}
\safemath{\noiseSymbol}{Z}
\newcommand{\noise[1]}{\noiseSymbol_{#1}}
\safemath{\noisePower}{\sigma^2_w}
\newcommand{\AFconstant[1]}{\alpha_{#1}}

\safemath{\entropy}{\mathrm{H}}
\safemath{\diffentropy}{h}

\newcommand{\information}{\mathrm{I}}
\newcommand{\expectation[1]}{\Exop\{#1\}}

\safemath{\corrsrSymbol}{\rho}
\newcommand{\corrsr[2]}{\corrsrSymbol_{{#1}{#2}}}
\safemath{\corrrrSymbol}{\delta}

\safemath{\corrssSymbol}{\lambda}

\safemath{\varSpaceSymbol}{\mathcal{X}}
\newcommand{\varSpace[1]}{\varSpaceSymbol_{#1}}

\safemath{\auxiliaryvarSpace}{\mathcal{Z}}

\safemath{\othervarSpaceSymbol}{\mathcal{Y}}
\newcommand{\othervarSpace[1]}{\othervarSpaceSymbol_{#1}}
\safemath{\estimationvarSpaceSymbol}{\widehat{\mathcal{Y}}}

\safemath{\varSymbol}{X}
\safemath{\othervarSymbol}{Y}
\safemath{\auxvarSymbol}{\widehat{Y}}
\safemath{\auxiliaryvar}{Z}

\safemath{\estimationvarSymbol}{\widehat{\othervarSymbol}}
\newcommand{\var[1]}{\varSymbol_{#1}}
\newcommand{\othervar[1]}{\othervarSymbol_{#1}}

\safemath{\auxvarsampleSymbol}{z}

\safemath{\varsetSymbol}{\boldsymbol{X}}
\safemath{\othervarsetSymbol}{\boldsymbol{Y}}
\safemath{\estimationvarsetSymbol}{\widehat{\othervarsetSymbol}}
\newcommand{\varset[1]}{\varsetSymbol_{#1}}
\newcommand{\othervarset[1]}{\othervarsetSymbol_{#1}}

\safemath{\varsampleSymbol}{x}
\safemath{\othervarsampleSymbol}{y}
\safemath{\auxiliaryvarsample}{z}

\safemath{\estimationvarsampleSymbol}{\hat{\othervarsampleSymbol}}

\newcommand{\varsample[1]}{\varsampleSymbol_{#1}}
\newcommand{\othervarsample[1]}{\othervarsampleSymbol_{#1}}

\safemath{\varsetsampleSymbol}{\boldsymbol{x}}
\newcommand{\varsetsample[1]}{\varsetsampleSymbol_{#1}}
\safemath{\othervarsetsampleSymbol}{\boldsymbol{y}}

\safemath{\estimationvarsetsampleSymbol}{\hat{\othervarsampleSymbol}}

\safemath{\varSequence}{X}
\safemath{\estimationvarSequence}{\widehat{X}}
\safemath{\othervarSequence}{Y}
\safemath{\auxiliaryvarSequence}{Z}

\safemath{\distortion}{d}
\safemath{\distortionLimit}{D}

\safemath{\rateSymbol}{R}
\safemath{\estimationrateSymbol}{\widehat{R}}
\safemath{\assistrateSymbol}{\tilde{R}}

\newcommand{\rate[1]}{\rateSymbol_{#1}}

\safemath{\blockindex}{b}
\safemath{\blocknumber}{B}

\safemath{\signalsetSymbol}{\mathcal{W}}
\newcommand{\signalset[1]}{\signalsetSymbol_{#1}}

\safemath{\signalnumber}{n}
\safemath{\signalSymbol}{w}
\safemath{\signalsSymbol}{\boldsymbol{w}}
\safemath{\estimationparamSymbol}{k}
\safemath{\assistparamSymbol}{z}

\newcommand{\signal[2]}{\signalSymbol_{#1}^{#2}}

\newcommand{\destsignal[2]}{\hat{\signalSymbol}_{#1}^{#2}}

\safemath{\gaussian}{\mathcal{N}}

\begin{document}

\title{\LARGE{Upper{-}Bounding the Capacity of Relay Communications{ - }Part I}}

\author{
  \IEEEauthorblockN{Farshad Shams}
  \IEEEauthorblockA{Dep. of Computer Science and Engineering\\
    IMT Institute for Advanced Studies Lucca, Italy\\
    Email: f.shams@imtlucca.it}
  \and
  \IEEEauthorblockN{Marco Luise}
  \IEEEauthorblockA{Dipartimento di Ingegneria dell'Informazione\\
    University of Pisa, Italy\\
    Email: marco.luise@iet.unipi.it}
}

%%%\author{Farshad Shams\thanks{F. Shams is with the Dept. Computer Science and Engineering,
%%%    IMT Institute for Advanced Studies, Piazza San Ponziano, 6, 55100, Lucca, Italy (email: f.shams@imtlucca.it).}
%%%  and Marco Luise \thanks{M. Luise is with the Dipartimento di Ingegneria dell'Informazione,
%%%    University of Pisa, Via G. Caruso, 16, 56122 Pisa, Italy (email: marco.luise@iet.unipi.it).}}
\maketitle

\begin{abstract}
This paper focuses on the capacity of point{-}to{-}point relay communications wherein the transmitter is assisted by an intermediate relay. We detail the mathematical model of cutset and amplify and forward (AF) relaying strategy. %%The cutset upper bound capacity is presented as a reference to compare another realistic strategy.
We present the upper bound capacity of each relaying strategy from information theory viewpoint and also in networks with Gaussian channels.
We exemplify various outer region capacities of the addressed strategies with two different case studies. The results exhibit that in low signal{-}to{-}noise ratio (SNR) environments the cutset performance is better than amplify and forward strategy.
\end{abstract}

\section{Introduction}
Cooperative communication enables single{-}antenna mobiles in a multi{-}user environment to share their antennas and generate a virtual multiple{-}antenna transmitter, and consequently exploit some of the benefits of multiple{-}input multiple{-}output (MIMO) systems. Cooperative transmission can increase the data rate, save transmission power, and extend the coverage range of the network. As a result, it is considered to be a key{-}technique in the development of a robust and efficient communication system \cite{Nosratinia04}. The first idea of cooperative transmission can be traced back to the proposal of the relay channel model, which consists of one source, one destination and one relay.
A relay channel models transmissions consists of a pair of transceivers communicate assisted by one intermediate node. To understand how much performance improvement can be obtained by a cooperative network, we use an information theoretical study. Such a study also results how reliable communications should take place in future wireless communications. To that end, we study the information theoretic aspects around peer{-}to{-}peer relayed communications. We calculate the outer region bound of cutset and AF relaying strategy.

%%%The three{-}terminal relay channel was introduced in 1971 by Van{-}der Meulen \cite{Meulen71} and was initially investigated in the context of information theory. In his seminal paper, Van{-}der Meulen introduced upper and lower bounds for the capacity of a relay channel.
%%%Meanwhile, Sato \cite{Sato76} also looked at the relay channel in the framework of the ALOHA protocol.
%%%But, the booming interest in cooperative communications was initiated by the paper of Cover and El Gamal \cite{CoverGamal79}. They evaluated the improvement of upper bound channel capacity of a point{-}to{-}point connection, placing one assistant node as relay and applying two different coding strategies of block Markov (\emph{decode and forward}) and side information encoding (\emph{compress and forward}). They also mixed these two strategies in \cite[Th. 7]{CoverGamal79}. Another relaying strategy proposed in the literature is \emph{amplify and forward} \cite{Laneman04} which augments the received signal without generating a new code at the relay.

The rest of this contribution is structured as follows. In \sectionname~\ref{sec:relay111}, we study relay assisted communication between a transmitter and a far destination which is out of its transmission range. Then, we extend the first scenario to a network wherein the direct{-}link between the transmitter and destination is available. We present an upper bound on the cutset capacity of relay communications in \sectionname~\ref{sec:cutset111}. \sectionname~\ref{sec:AFTech111} is devoted to the well{-}known relaying strategies: the amplify and forward (AF). We formally define these relaying protocols and calculate the upper bound capacity for each of them. We illustrate our results in \sectionname~\ref{sec:casestudy111}, and then conclude in \sectionname~\ref{sec:relay111dis}.

\begin{notation}
%%%We use the same notation as that in \cite{Gallager68}.
%%%Some general notation is needed in the following. 
Upper case letters $\var[i], \othervar[i]$ represent the output and input random vector variables of node $i$, respectively, and $\var[i][\blockindex]$ is the $\blockindex${th} entry of a random vector. 
%%%The notations $\Varop\left({\var[]}\right)$, $\expectation[{\var[]}]$, $\entropy\left(\var[]\right)$, and $\information\left(\var[]; \othervar[]\right)$ are used to denote variance, expectation value, entropy, and mutual information, respectively. 
%%%For the sake of simplicity we will assume that our channels are real{-}valued.
The notation $\pathgain[i]{j}$ is used for the real{-}valued channel gain of the link between nodes $i$ and $j$. The symbol sent by transmitter $\sourceindex$ in block index $\blockindex$ is represented by $\signal[\sourceindex]{\blockindex}$. We use calligraphic letters $\varSpace[], \othervarSpace[], \signalset[]$ to indicate (finite) alphabet spaces and lower case letters $\varsample[i], \othervarsample[i]$ for channel distributions. 
The conventions $\varset[\setname]$ and $\varsetsample[\setname]$ denote joint vector variable and channel distribution, respectively of the indices belonging to set $\setname$. $\setname^C$ represents the complementary set of $\setname$.
%%%The sign $\widehat{\text{ }}$ is used for estimated parameters, and the notation $\distortion(\var[], \widehat{\var[]})$ is the average distortion measure between two sequences.
The notation $\capacitySymbol(\SNR[i]{j})=\frac{1}{2}\log_{2}\left(1+\SNR[i]{j}\right)$ denotes Shannon channel capacity between nodes $i$ and $j$ with SNR $\SNR[i]{j}$ at the receiver.  The symbol $\upperboundSymbol$ is used for the outer region channel capacity.
%%%Moreover, we use the symbol $A_{\epsilon}^{\signalnumber}$ to denote a strongly typical set as defined in \cite[p. 59]{ThomasCoverBook}.
To compute (conditional) mutual information, we need the statistical parameter:
\begin{equation}\label{eq:detcov}
\covs\lefto[\varset[\setA]|\othervarset[\setB]\right]\triangleq\det\{\Sigma_{\varset[\setA]|\othervarset[\setB]}\}\\
\end{equation}
wherein $\det\{\}$ denotes the determinant of a matrix and
\begin{subequations}
\begin{gather}
\Sigma_{\varset[\setA]|\othervarset[\setB]} \,=  \,\Sigma_{\varset[\setA]\varset[\setA]} - \Sigma_{\varset[\setA]\othervarset[\setB]}\cdot\inv{\left(\Sigma_{\othervarset[\setB]\othervarset[\setB]}\right)}\cdot\Sigma_{\othervarset[\setB]\varset[\setA]}\\
\Sigma = \left[ \begin{array}{ccc} \Sigma_{\varset[\setA]\varset[\setA]} & \Sigma_{\varset[\setA]\othervarset[\setB]} \\
\Sigma_{\othervarset[\setB]\varset[\setA]} & \Sigma_{\othervarset[\setB]\othervarset[\setB]} \end{array}\right]
\end{gather}
\end{subequations}
wherein $\Sigma$ denotes the covariance matrix of multivariate Gaussian distribution $\varset[\setA]$ jointly with $\othervarset[\setB]$ \cite[Ch. 5]{Allan09}.
\hfill$\blacksquare$
\end{notation}

\begin{notation}\label{nota:scaledpathgain}
The power gain of the (real{-}valued) radio link between two nodes $i$ and $j$ at a distance $d_{ij}$, is scaled by
\begin{equation}
\Gamma_{ij}=G_{tx}\cdot G_{rx}\cdot\left(\frac{\lambda_{0}}{4\pi}\right)^{2}\cdot d^{-\alpha}_{ij}
\end{equation}
wherein $G_{tx}=G_{rx}=1$ are the transmit and receive antennas gains assumed omnidirectional, the parameter $\lambda_{0}\!\approxeq\!0.12\text{m}$ represents carrier wave length, and the path loss exponent is $\alpha=2$. Consequently, if the node $i$ is placed close to node $j$ and far from node $k$ then $\pathgain[i]{j}\!\gg\!\pathgain[i]{k}$. \hfill$\blacksquare$
\end{notation}

\section{Relayed point{-}to{-}point communication}\label{sec:relay111}
The relays are able to outperform a direct{-}link communication and even actualize some otherwise impossible scenario. For example,
\figurename~\ref{fig:fardest} depicts the simplest cooperative (relay) communication model to realize a communication between two hidden terminals. The source $\sourceindex$ wishes to send a message to the far destination $\destindex$ which is out of sight.  We assume there is no fading on the wireless channels. One possible solution is to use an intermediate node. The source $\sourceindex$ sends a message to the relay $\relayindex$ and the received noisy version of the original message is re{-}transmitted to the far destination node. We assume that the relay is accessible to both the source and destination nodes. In the example, the relay does not do any processing (e.g. encoding, decoding) on the receive signal, but its duty is to make possible the information exchange from $\sourceindex$ to $\destindex$. This simple two{-}hop model enlarges significantly the range of the network.
\begin{figure}
  \begin{center}
    \psfrag{s}[r][b][0.9]{$\sourceindex$}
    \psfrag{r}[c][b][0.9]{$\relayindex$}
    \psfrag{d}[l][b][0.9]{$\destindex$}
    \psfrag{Zr}[c][b][0.85]{$\noise[\relayindex]$}
    \psfrag{Zd}[c][b][0.85]{$\noise[\destindex]$}
    \psfrag{hsr}[c][b][0.85]{$\sqrt{\pathgain[\sourceindex]{\relayindex}}$}
    \psfrag{hrd}[c][b][0.85]{${\sqrt{\pathgain[\relayindex]{\destindex}}}$}
    \includegraphics[width=0.7\columnwidth]{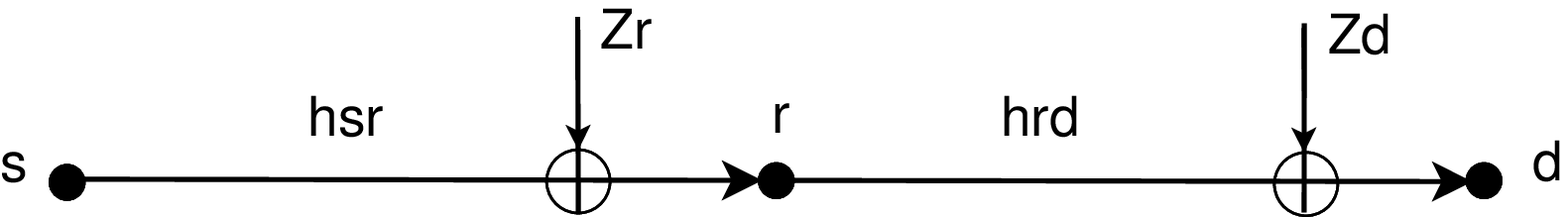}
  \end{center}
  \caption{One source ($\sourceindex$), one relay ($\relayindex$), one far destination ($\destindex$) network.}
  \label{fig:fardest}
\end{figure}
The relaying function between the source and destination nodes can be done in two different signaling modes: \emph{half{-}duplex} and \emph{full{-}duplex}. In half{-}duplex mode, the transmission from source to destination is done in two different stages and in each stage the relay acts either as receiver or transmitter. In the first stage, $\sourceindex$ transmits a stream of information and $\relayindex$ operates as a receiver and $\destindex$ is idle. In the second stage, the channel between $\sourceindex$ and $\relayindex$ is kept idle and $\destindex$ is active to receive data from the intermediate node $\relayindex$ which acts only as a transmitter. On the other hand, in full{-}duplex mode, both wireless channels are simultaneously busy and the relay plays the role of a receiver and a transmitter at the same time. The function of the intermediate node is feasible using two distinct antennas, one as receiver and one as transmitter using two orthogonal frequency band. In spite of the the difficulties to apply the full{-}duplex manner in wireless networks, in this work, we concern with this model for didactic issues. Our goal is to determine how much data we can reliably get from source to destination, placing no importance on delay.%% or computational complexity.

In the scenario of \figurename~\ref{fig:fardest}, the power expenditure of $\sourceindex$ and $\relayindex$, $\power[\sourceindex]{}$ and $\power[\relayindex]{}$, respectively can be set a priori or can be adjusted given the individual power constraints $\powermax[\sourceindex]{}$ and $\powermax[\relayindex]{}$ and to the path gain values $\pathgain[\sourceindex]{\relayindex}$ and $\pathgain[\relayindex]{\destindex}$. The objective of the power control is to approach the maximum Shannon channel capacity between source and a (far) destination. According to the max{-}flow min{-}cut theorem (also referred to as the cutset bound), the maximum end{-}to{-}end channel capacity in \figurename~\ref{fig:fardest} is achieved when the capacity of the source{-}relay link is the same as that of the $\relayindex\to\destindex$ link, i.e.
%%%\begin{equation}
$\capacitySymbol(\SNR[\sourceindex]{\relayindex})=\capacitySymbol(\SNR[\relayindex]{\destindex})$. 
%%%\end{equation}
We assume that the noises $\noise[\relayindex]$ and $\noise[\destindex]$ are Gaussian random variables $\gaussian(0, \noisePower)$. The power optimization problem becomes to the equation:
\begin{equation}\label{eq:fardest}
\pathgain[\sourceindex]{\relayindex}\power[\sourceindex]{}=\pathgain[\relayindex]{\destindex}\power[\relayindex]{}
\end{equation}

We study the final equation in four different cases:
\begin{enumerate}
  \item if $(\pathgain[\sourceindex]{\relayindex}\!\le\!\pathgain[\relayindex]{\destindex})$ and $(\powermax[\sourceindex]{}< \powermax[\relayindex]{})$ then\;$\power[\sourceindex]{}=\powermax[\sourceindex]{},\; \power[\relayindex]{}=\displaystyle\frac{\pathgain[\sourceindex]{\relayindex}}{\pathgain[\relayindex]{\destindex}}\powermax[\relayindex]{}$~;
  \item if $(\pathgain[\sourceindex]{\relayindex}\!<\!\pathgain[\relayindex]{\destindex})$ and $(\powermax[\sourceindex]{}\ge \powermax[\relayindex]{})$ then\;$\power[\sourceindex]{}=\powermax[\relayindex]{},\; \power[\relayindex]{}=\displaystyle\frac{\pathgain[\sourceindex]{\relayindex}}{\pathgain[\relayindex]{\destindex}}\powermax[\relayindex]{}$~;
  \item if $(\pathgain[\sourceindex]{\relayindex}\!\ge\!\pathgain[\relayindex]{\destindex})$ and $(\powermax[\sourceindex]{}> \powermax[\relayindex]{})$ then\; $\power[\sourceindex]{}=\displaystyle\frac{\pathgain[\relayindex]{\destindex}}{\pathgain[\sourceindex]{\relayindex}}\powermax[\sourceindex]{},\; \power[\relayindex]{}=\powermax[\relayindex]{}$~;
  \item if $(\pathgain[\sourceindex]{\relayindex}\!>\!\pathgain[\relayindex]{\destindex})$ and $(\powermax[\sourceindex]{}\le \powermax[\relayindex]{})$ then\; $\power[\sourceindex]{}=\displaystyle\frac{\pathgain[\relayindex]{\destindex}}{\pathgain[\sourceindex]{\relayindex}}\powermax[\sourceindex]{},\; \power[\relayindex]{}=\powermax[\sourceindex]{}$~.
\end{enumerate}

Case (1) is the situation in which the maximum capacity of the source{-}relay link is less than that relay{-}destination link. The source exploits the maximum capacity of the source{-}relay link, i.e. $\power[\sourceindex]{}=\powermax[\sourceindex]{}$. To adjust $\power[\relayindex]{}$ it is enough to satisfy \equationname~\ref{eq:fardest}, since the relay{-}destination channel capacity is limited to the capacity of source{-}relay link. The same derivation applies to case (3). In case (2) there is no exact relation between the maximum channel capacity of the two channels. Assigning $\power[\sourceindex]{}=\powermax[\relayindex]{}$, the source{-}relay capacity is bounded by $\pathgain[\sourceindex]{\relayindex}\powermax[\relayindex]{}$ that is less than the maximum capacity of the source{-}relay link. At this point it is easy to satisfy \equationname~\ref{eq:fardest}. Case (4) follows from (2).

Cooperative communication can be efficient also when a direct{-}link between the source node $\sourceindex$ and destination $\destindex$ is available. In the network illustrated by \figurename~\ref{fig:onesonerencdec}, when the source node $\sourceindex$ broadcasts, a noisy version of the data comes to the relay $\relayindex$ and another corrupted version of data approaches the destination $\destindex$. Using an intermediate node is useful when the received data at $\destindex$ is too weak to be decoded. In this situation, a relay can help communication by transmitting a new version of its own received signal. The direct{-}link signal is used to help decoding the stronger version of original data sent by relay.
\begin{figure}
  \begin{center}
    \psfrag{s}[r][b]{$\sourceindex$}
    \psfrag{r}[c][b]{$\relayindex$}
    \psfrag{d}[l][b]{$\destindex$}
    \psfrag{Enc}[c][b][0.7]{Encoder}
    \psfrag{Dec}[c][b][0.7]{Decoder}
    \psfrag{Zr}[c][b][0.75][48]{$\noise[\relayindex]$}
    \psfrag{Zd}[c][b][0.75]{$\noise[\destindex]$}
    \psfrag{hsr}[c][b][0.75][45]{$\sqrt{\pathgain[\sourceindex]{\relayindex}}$}
    \psfrag{hrd}[c][b][0.75][-35]{$\sqrt{\pathgain[\relayindex]{\destindex}}$}
    \psfrag{hsd}[c][b][0.75]{$\sqrt{\pathgain[\sourceindex]{\destindex}}$}
    \psfrag{xs}[c][b][0.8]{$\var[\sourceindex]$}
    \psfrag{yr}[c][b][0.8]{$\othervar[\relayindex]$}
    \psfrag{xr}[c][b][0.8]{$\var[\relayindex]$}
    \psfrag{yd}[c][b][0.8]{$\othervar[\destindex]$}
    \psfrag{ws}[c][b][0.75]{$\signal[\sourceindex]{}$}
    \psfrag{wd}[c][b][0.75]{$\destsignal[\sourceindex]{}$}
    \includegraphics[width=0.8\columnwidth]{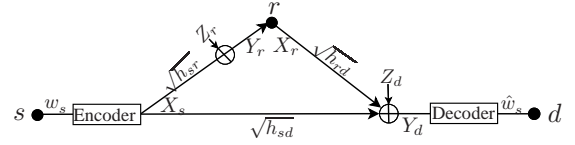}
  \end{center}
  \caption{One source, one relay, one destination network scenario.}
  \label{fig:onesonerencdec}
\end{figure}

The cooperative network in \figurename~\ref{fig:onesonerencdec} consists of four finite random spaces: $\varSpace[\sourceindex]$ at the transmitter, $\othervarSpace[\relayindex]$ and $\varSpace[\relayindex]$ at the relay node, and $\othervarSpace[\destindex]$ at the destination. %%The conditional probability mass function (pmf) of the random variables is described as $p(\othervarsample[\destindex]\, \othervarsample[\relayindex] | \varsample[\relayindex]\,\varsample[\sourceindex])$ .
The source node wants to transmit a message $\signal[\sourceindex]{}$ to the destination through direct and reliable links. The original message $\signal[\sourceindex]{}$ is split in a sequence of sub{-}messages $\signal[\sourceindex]{1},\dots,\signal[\sourceindex]{\blockindex},\dots,\signal[\sourceindex]{\blocknumber}$ each uniformly and independently drawn from a set with alphabet size $\signalnumber$ and length $\rate[\sourceindex]{}$, represented by ${\signalset[\sourceindex]=\{0, 1, ..., 2^{\signalnumber\rate[\sourceindex]{}}-1\}}$. The encoder at the transmitter is a function $\signalset[\sourceindex]\to\varSpace[\sourceindex]$ which maps $\signal[\sourceindex]{\blockindex}$ to $\var[\sourceindex][\blockindex]\left(\signal[\sourceindex]{\blockindex}\right)$.
We assume that the decoder generates a sequence of $\blocknumber$ Gaussian random codewords given a constraint on average power as $\expectation[{\var[\sourceindex]^{2}}]\le\powermax[\sourceindex]{}$. We also assume $\expectation[{\var[\sourceindex]}]=0$.
In each block with index \blockindex, the source $\sourceindex$ broadcasts the encoded symbols and each $\var[\sourceindex][\blockindex]\left(\signal[\sourceindex]{\blockindex}\right)$ experiences two different paths to approach the relay and destination. The relay observes $\othervar[\relayindex][\blockindex]$ as:
\begin{align}\label{eq:srbroadcast}
\othervar[\relayindex][\blockindex] = \sqrt{\pathgain[\sourceindex]{\relayindex}}\var[\sourceindex][\blockindex]\left(\signal[\sourceindex]{\blockindex}\right) + \noise[\relayindex]
\end{align}

The receive signals at $\relayindex$ and $\destindex$ are different but statistically correlated. In full{-}duplex mode, the relay processes the received signal in the previous block index and generates the information $\var[\relayindex][\blockindex]\left(\signal[\sourceindex]{\blockindex-1}\right)$, to be sent into the {relay{-}destination} channel. The information $\var[\relayindex][\blockindex]\left(\signal[\sourceindex]{\blockindex-1}\right)$ is a re{-}generated version of $\othervar[\relayindex][\blockindex-1]$ and it is the output of the relay $\relayindex$'s deterministic function whose input is the sequence of the previous received signals:
\begin{align}\label{eq:deterministinfunc}
\var[\relayindex][\blockindex]\left(\signal[\sourceindex]{\blockindex-1}\right) = \relayfunction[\relayindex]{\blockindex}(\othervar[\relayindex][\blockindex-1], \othervar[\relayindex][\blockindex - 2], \dots,\othervar[\relayindex][1]);
\end{align}
The function $\relayfunction[]{}$ depends on the specific cooperative strategy as will be specified later on.
Each symbol sequence $\{\var[\relayindex][\blockindex]\left(\signal[\sourceindex]{\blockindex-1}\right)\}$ is such that $\expectation[{\var[\relayindex]^{2}}]\le\powermax[\relayindex]{}$, and $\expectation[{\var[\relayindex]}]=0$.
The destination receives a superposition of two different signals:
\begin{align}\label{eq:destreceive}
\othervar[\destindex][\blockindex] = \sqrt{\pathgain[\sourceindex]{\destindex}}\var[\sourceindex][\blockindex]\left(\signal[\sourceindex]{\blockindex}\right) + \sqrt{\pathgain[\relayindex]{\destindex}}\var[\relayindex][\blockindex]\left(\signal[\sourceindex]{\blockindex-1}\right)+ \noise[\destindex]
\end{align}

It is seen that, the destination node receives information both about $\signal[\sourceindex]{\blockindex}$ and $\signal[\sourceindex]{\blockindex-1}$. It means, the destination receives two different versions of $\signal[\sourceindex]{\blockindex}$ in two broadcast and multiple-access stages. The decoder at the destination decodes the message $\destsignal[\sourceindex]{}$ canceling the effect of $\var[\sourceindex][\blockindex-1]\left(\signal[\sourceindex]{\blockindex-1}\right)$ from $\var[\relayindex][\blockindex]\left(\signal[\sourceindex]{\blockindex-1}\right)$. In particular, the decoding function is a mapping function from $\othervarSpace[\destindex]\rightarrow\displaystyle\signalset[\sourceindex]$.
The error probability at the destination's decoder is defined as:
\begin{align}
P_{e}^{\signalnumber}=\displaystyle{2^{-\signalnumber\rate[\sourceindex]{}}}\cdot\sum_{\signal[\sourceindex]{}}{\text{Pr }\{\destsignal[\sourceindex]{}\neq\signal[\sourceindex]{}|\signal[\sourceindex]{}\,\text{was sent}\}}
\end{align}
which is defined based on the assumption that the messages are independent and uniformly distributed over the alphabet space. The rate $\rate[\sourceindex]{}$ is achievable if there exists a sequence of codes $\left(\signalnumber,2^{\signalnumber\rate[\sourceindex]{}}\right)$ for which $P_e^{\signalnumber\to\infty}$ is arbitrarily close to zero.

%%In a full{-}duplex relay communication, the relay has two different antennas to communicate to different links. The output signals of source and relay nodes can be accommodated on either one or two antennas at the destination. In a reliable communication, the time signaling coordination between source and relay is necessary. For this aim, we introduce a global network trigger that handles the transmission synchronization of $\var[\sourceindex][\blockindex]\left(\signal[\sourceindex]{\blockindex}\right)$ and $\var[\relayindex][\blockindex]\left(\signal[\sourceindex]{\blockindex-1}\right)$.

%%In the sequel we introduce different strategies of the re{-}generation code at the relay and decoding function at the destination and their achievable upper bound capacity from information theoretical and AWGN point of views.
%%In the sequel we introduce different relaying strategies and their achievable upper bound capacity from information theoretical and AWGN point of views.

\section{Cutset upper bound}\label{sec:cutset111}

In this section we derive the upper bound capacity of max{-}flow min{-}cut or ``cutset". The cutset upper bound is used as a reference to compare the upper bound of the realistic models. M. R. Aref in \cite[Th. 3.4]{aref-thesis} pioneered to establish the cutset bound in a general reliable network with multiple relays. S. Zahedi in \cite[Th. 2.2]{zahedi-thesis} presents another proof to the cutset upper bound in the one relay case like \figurename~\ref{fig:onesonerencdec}.
The following proposition shows that a cooperative system can be decomposed into a broadcast channel from the source node's viewpoint, and a multiple{-}access channel from the destination point of view.
\begin{proposition}\label{pr:cutsetinf}
\cite[Th. 2.2]{zahedi-thesis} For any relay channel $(\varSpace[\sourceindex]\times\varSpace[\relayindex], p(\othervarsample[\destindex]\,\othervarsample[\relayindex]\,|\,\varsample[\sourceindex]\,\varsample[\relayindex]), \othervarSpace[\relayindex]\times\othervarSpace[\destindex])$ the cutset capacity is upper bounded by
\begin{equation*}
\upperbound[cutset]{(\sourceindex;\relayindex;\destindex)}\!=\!
\underset{p(\varsample[\sourceindex], \varsample[\relayindex])}{\sup}\min\!\left\{\!\information(\var[\sourceindex];\, \othervar[\destindex]\,\othervar[\relayindex] | \var[\relayindex]),\, \information(\var[\sourceindex]\,\var[\relayindex];\,\othervar[\destindex])\!\right\}
\end{equation*}
where the supremum is computed over all joint distributions on $\varSpace[\sourceindex]\times\varSpace[\relayindex]$ complying with individual power constraints.
\end{proposition}

The first term is the mutual information of broadcasting $\var[\sourceindex]$ toward $\relayindex$ and $\destindex$ with transition probability $p(\othervarsample[\destindex]\,\othervarsample[\relayindex]|\,\varsample[\sourceindex])$. The second term is the mutual information of multiple{-}access of $\relayindex$ and $\sourceindex$ at the destination node with transition probability $p(\othervarsample[\destindex]\,|\,\varsample[\sourceindex]\,\varsample[\relayindex])$. Thus, in general, random variables $\othervarsample[\destindex]$ and $\othervarsample[\relayindex]$ are statistically related to both inputs $\varsample[\sourceindex]$ and $\varsample[\relayindex]$ through $p(\othervarsample[\destindex]\,\othervarsample[\relayindex]\,|\,\varsample[\sourceindex]\,\varsample[\relayindex])$.

$p(\varsample[\sourceindex], \varsample[\relayindex])$ is a joint Gaussian distribution on $\varSpace[\sourceindex]\times\varSpace[\relayindex]$ with a cross correlation coefficient of $\displaystyle\corrsr[]{}=\frac{\expectation[{\var[\sourceindex]\var[\relayindex]}]}{\sqrt{\expectation[{\var[\sourceindex]^{2}}]\expectation[{\var[\relayindex]^{2}}]}}$. In the case of $\corrsr[]{}=0$ we have: {$p(\varsample[\sourceindex], \varsample[\relayindex]) = p(\varsample[\sourceindex])\cdot p(\varsample[\relayindex])$}.

At this point, We recall some useful statistics equalities related to \figurename~\ref{fig:onesonerencdec}.
To review the algebraic manipulation of the following formulas, see \cite[Ch. 5]{Allan09}.
\begin{subequations}\label{eqs:onerelaystatistics}
\begin{flalign}
&\Varop\left(\othervar[\destindex]\right)=\pathgain[\sourceindex]{\destindex}\powermax[\sourceindex]{}+\pathgain[\relayindex]{\destindex}\powermax[\relayindex]{}+2\corrsr[]{}\sqrt{\pathgain[\sourceindex]{\destindex}\powermax[\sourceindex]{}\pathgain[\relayindex]{\destindex}\powermax[\relayindex]{}}+\noisePower &\\
%%&\Varop\left(\var[\sourceindex] | \var[\relayindex]\right)= \powermax[\sourceindex]{}\left(1-\corrsr[]{}^2\right) \\
%%&\Varop\left(\othervar[\relayindex] | \var[\relayindex]\right)= \pathgain[\sourceindex]{\relayindex}\powermax[\sourceindex]{}\left(1-\corrsr[]{}^2\right) + \noisePower\\
%%&\Covop\left(\othervar[\destindex] , \var[\relayindex]\right)= \corrsr[]{}\sqrt{\pathgain[\sourceindex]{\destindex}\powermax[\sourceindex]{}\powermax[\relayindex]{}} + \sqrt{\pathgain[\relayindex]{\destindex}}\,\powermax[\relayindex]{}\\
%%&\Varop\left(\othervar[\destindex]|\var[\relayindex]\right) = \pathgain[\sourceindex]{\destindex}\powermax[\sourceindex]{}\left(1-\corrsr[]{}^2\right)+\noisePower\\
&\;\covs\lefto[\othervar[\destindex]\othervar[\relayindex]|\var[\relayindex]\right] = \left(\pathgain[\sourceindex]{\destindex} + \pathgain[\sourceindex]{\relayindex}\right)\powermax[\sourceindex]{}\left(1-\corrsr[]{}^2\right)+\noisePower\\
&\;\covs\lefto[\othervar[\destindex]\othervar[\relayindex]|\var[\sourceindex]\var[\relayindex]\right] = \;\covs\lefto[\othervar[\destindex]|\var[\sourceindex]\var[\relayindex]\right] =\noisePower
\end{flalign}
\end{subequations}
%%\;\mbox{(see \eqref{eq:detcov})}
Reference \cite[Proposition 2]{Kramer05} shows that the $\upperbound[cutset]{(\sourceindex;\relayindex;\destindex)}$ is attained by Gaussian channels. The following theorem presents the capacity of the AWGN relay channel.
\begin{theorem}\label{th:AWGNcutset}
The AWGN cutset capacity of a point{-}to{-}point relayed communication is upper bounded by:
\begin{multline}\label{eq:AWGNcutset}
\upperbound[cutset]{(\sourceindex;\relayindex;\destindex)}\!=\!\underset{-1\le\corrsr[]{}\le1}{\sup}\min \displaystyle\left\{\capacitySymbol\left(\frac{\left(\pathgain[\sourceindex]{\destindex}\!+\! \pathgain[\sourceindex]{\relayindex}\right)\powermax[\sourceindex]{}\left(1\!-\!\corrsr[]{}^2\right)}{\noisePower}\right) \right. \\
\qquad,
\left.\capacitySymbol\left(\frac{\pathgain[\sourceindex]{\destindex}\powermax[\sourceindex]{}+\pathgain[\relayindex]{\destindex}\powermax[\relayindex]{}+2\corrsr[]{}\sqrt{\pathgain[\sourceindex]{\destindex}\powermax[\sourceindex]{}\pathgain[\relayindex]{\destindex}\powermax[\relayindex]{}}}{\noisePower}\right) \right\}
\end{multline}
\end{theorem}
\begin{proof} The mutual information terms are calculated using \equationsname~\ref{eqs:onerelaystatistics}.
\begin{flalign*}
&\information(\var[\sourceindex]\,;\,\othervar[\destindex]\,\othervar[\relayindex] | \var[\relayindex])\!=\!\displaystyle\frac{1}{2}\log_2{\frac{\covs\lefto[\othervar[\destindex]\othervar[\relayindex]|\var[\relayindex]\right]}{\covs\lefto[\othervar[\destindex]\othervar[\relayindex]|\var[\relayindex]\var[\sourceindex]\right]}}
=&\\
&\qquad\!=\!\frac{1}{2}\log_2{\frac{\covs\lefto[\othervar[\destindex]\othervar[\relayindex]|\var[\relayindex]\right]}{\noisePower}}
=\capacitySymbol\left(\frac{\left(\pathgain[\sourceindex]{\destindex} + \pathgain[\sourceindex]{\relayindex}\right)\powermax[\sourceindex]{}\left(1-\corrsr[]{}^2\right)}{\noisePower}\right);\\
%%\end{multline*}
%%\begin{multline*}
&\displaystyle\information(\var[\sourceindex]\,\var[\relayindex]\,;\,\othervar[\destindex]) = \frac{1}{2}\log_2{\frac{\Varop\left(\othervar[\destindex]\right)}{\covs\lefto[\othervar[\destindex]|\var[\sourceindex]\,\var[\relayindex]\right]}} =\frac{1}{2}\log_2{\frac{\Varop\left(\othervar[\destindex]\right)}{\noisePower}}=\\
&\qquad=\capacitySymbol\left(\frac{\pathgain[\sourceindex]{\destindex}\powermax[\sourceindex]{}+\pathgain[\relayindex]{\destindex}\powermax[\relayindex]{}+2\corrsr[]{}\sqrt{\pathgain[\sourceindex]{\destindex}\powermax[\sourceindex]{}\pathgain[\relayindex]{\destindex}\powermax[\relayindex]{}}}{\noisePower}\right). \end{flalign*}
\end{proof}
Performing the maximization over $\corrsr[]{}$, we can easily obtain the upper bound. In other words, under average power constraints, a jointly Gaussian input distribution simultaneously maximizes both mutual information terms. The important result of
\theoremname~\ref{th:AWGNcutset} is: If $\SNR[\sourceindex]{\relayindex}>\SNR[\relayindex]{\destindex}$ (i.e. the relay node is in a good position to receive signals from the transmitter rather than to deliver symbols to the destination) then the cutset upper bound capacity is achieved by the maximization of the multiple{-}access term (the second term), otherwise it is achieved by the broadcast term.

In the AWGN cutset upper bound capacity, when $\corrsr[]{}\ge0$, the broadcast term, $\information(\var[\sourceindex]\;;\; \othervar[\destindex]\,\othervar[\relayindex] | \var[\relayindex])$, is decreasing function of $\corrsr[]{}$, while the multiple{-}access term, $\information(\var[\sourceindex]\,\var[\relayindex]\;;\;\othervar[\destindex])$, is increasing. Thus, the maximum of $\upperbound[cutset]{(\sourceindex;\relayindex;\destindex)}$ is taken at the point at which two terms are equal:
\begin{flalign}\label{eq:maxeq}
\left(\pathgain[\sourceindex]{\relayindex}+\pathgain[\sourceindex]{\destindex}\right)\powermax[\sourceindex]{}\,\corrsr[]{}^{2} +2\sqrt{\pathgain[\sourceindex]{\destindex}\powermax[\sourceindex]{}\pathgain[\relayindex]{\destindex}\powermax[\relayindex]{}}\;\corrsr[]{}+\left(\pathgain[\relayindex]{\destindex}\powermax[\relayindex]{}-\pathgain[\sourceindex]{\relayindex}\powermax[\sourceindex]{}\right)=0
\end{flalign}

We analyze the \equationname~\ref{eq:maxeq} in different cases. First, we assume the parameter $\corrsr[]{}\ge0$ is fixed and it is possible for the source and relay nodes to adjust the transmission powers. Then, we assume the source and relay transmit at the maximum power and the network can tune the value of $\corrsr[]{}$.
\begin{enumerate}[1)]
\item~We assume $\varSpace[\sourceindex]$ and $\varSpace[\relayindex]$ are statistically independent; i.e. $\corrsr[]{}=0$. The capacity region of $\upperbound[cutset]{(\sourceindex;\relayindex;\destindex)}$ achieves its maximum with adjusting the power expenditure of the source and relay nodes. In \equationname~\ref{eq:maxeq} with $\corrsr[]{}=0$, it is enough to satisfy $\pathgain[\relayindex]{\destindex}\power[\relayindex]{}-\pathgain[\sourceindex]{\relayindex}\power[\sourceindex]{}=0$, that is equal to the maximum capacity problem of \figurename~\ref{fig:fardest} and \equationname~\ref{eq:fardest}. With fixed $\corrsr[]{}=0$ and an appropriate power control, the $\upperbound[cutset]{(\sourceindex;\relayindex;\destindex)}$ takes:
\begin{equation*}\label{eq:maxcutset}
\upperbound[cutset]{(\sourceindex;\relayindex;\destindex)}\!=\! \capacitySymbol\left(\frac{\left(\pathgain[\sourceindex]{\relayindex}\!+\!\pathgain[\sourceindex]{\destindex}\right)\power[\sourceindex]{}}{\noisePower}\right)\!=\!\capacitySymbol\left(\frac{\pathgain[\sourceindex]{\destindex}\power[\sourceindex]{}\!+\!\pathgain[\relayindex]{\destindex}\power[\relayindex]{}}{\noisePower}\right) \end{equation*}
The relation $\pathgain[\sourceindex]{\relayindex}\power[\sourceindex]{}\!=\!\pathgain[\relayindex]{\destindex}\power[\relayindex]{}$ means that the data rate of the channel between the source and relay nodes is equal to that between the relay and destination. In this case, the ${\sourceindex\to\relayindex\to\destindex}$ path has the maximum possible efficiency. Choosing two independent random spaces for $\varSpace[\sourceindex]$ and $\varSpace[\relayindex]$ guarantees minimum processing for relaying's data process. Suppose the destination node consists of two different antennas with orthogonal frequencies which are used simultaneously for receiving data from the source and relay nodes. Thus, there is no interference between the $\relayindex\to\destindex$ and $\sourceindex\to\destindex$ links. So, it results:
\begin{equation*}
\upperbound[cutset]{(\sourceindex;\relayindex;\destindex)} = \capacitySymbol\left(\frac{\pathgain[\sourceindex]{\destindex}\power[\sourceindex]{}+\pathgain[\relayindex]{\destindex}\power[\relayindex]{}}{\noisePower}\right)=\capacitySymbol\left(\SNR[\sourceindex]{\destindex}+\SNR[\relayindex]{\destindex}\right). \end{equation*}
Therefore, when $\corrsr[]{}=0$, the requirements to achieve the upper bound cutset capacity is an appropriate power control, and an adder component at destination to sum the received SNRs. Equivalently, the reliable communication forms a parallel channel between $\sourceindex\to\destindex$ and $\relayindex\to\destindex$ links.
%%and the $\SNR[\sourceindex]{\relayindex}$ is unimportant at all. This situation is useful when the source-relay link is very weak.
\item~If the sequences of $\var[\sourceindex]$ and $\othervar[\relayindex]$ are drawn from two correlated code spaces with a strictly positive correlation $\left(\corrsr[]{}>0\right)$, the necessary condition for the power control to achieve the capacity $\upperbound[cutset]{(\sourceindex;\relayindex;\destindex)}$ is: $\SNR[\sourceindex]{\relayindex}>\SNR[\relayindex]{\destindex}$. This means that the data rate of $\relayindex\to\destindex$ channel is less than that $\sourceindex\to\relayindex$ link. Hence, there must be a delay between sending the broadcast message and the multiple{-}access message.
\item~Finally, we suppose that the source and relay nodes transmit at maximum power and the network is able to tune the correlation parameter. The appropriate value of $\corrsr[]{}$ is found by resolving \equationname~\ref{eq:maxeq} for $\corrsr[]{}$.
\begin{equation*}\label{eq:bestcorrsr}
\corrsr[]{}^* = \frac{-\sqrt{\pathgain[\sourceindex]{\destindex}\powermax[\sourceindex]{}\pathgain[\relayindex]{\destindex}\powermax[\relayindex]{}}+\sqrt{\pathgain[\sourceindex]{\relayindex}\powermax[\sourceindex]{}\left(\pathgain[\sourceindex]{\destindex}\powermax[\sourceindex]{}\!+\!\pathgain[\sourceindex]{\relayindex}\powermax[\sourceindex]{}\!-\!\pathgain[\relayindex]{\destindex}\powermax[\relayindex]{}\right)}}{\left(\pathgain[\sourceindex]{\destindex}+\pathgain[\sourceindex]{\relayindex}\right)\powermax[\sourceindex]{}}
\end{equation*}
on the condition that $\Delta=\left(\pathgain[\sourceindex]{\destindex}+\pathgain[\sourceindex]{\relayindex}\right)\powermax[\sourceindex]{}-\pathgain[\relayindex]{\destindex}\powermax[\relayindex]{}> 0$ and $0\le\corrsr[]{}^*\le1$. If on the contrary $\Delta\!\le\!0$, then the capacity of $\relayindex\!\to\!\destindex$ is higher than that of the broadcast channel. This means that the link between relay and destination must be kept idle for receiving the broadcast message and therefore using the links is not highly efficient. Instead, the condition $\Delta\!>\!0$ means the broadcast capacity is higher than $\relayindex\!\to\!\destindex$ channel data rate. Using an appropriate memory at the relay node, all channels get busy.\\
\end{enumerate}

For strictly negative $\corrsr[]{}$, from the formula of $\upperbound[cutset]{(\sourceindex;\relayindex;\destindex)}$ in \theoremname~\ref{th:AWGNcutset}, it is derived that reducing coefficient $\corrsr[]{}$ toward $-1$ yields decreasing either broadcast and multiple{-}access capacities. The solution to compensate the affect of a negative $\corrsr[]{}$ is to raise significantly the upper bound limits of the source and relay power consumption.

%%For strictly negative $\corrsr[]{}$, from the formula of $\upperbound[cutset]{(\sourceindex;\relayindex;\destindex)}$ in theorem \ref{th:AWGNcutset}, it is followed that reducing $\corrsr[]{}$ toward $-1$ yields lessening the broadcast and multiple-access capacity. The solution to compensate the affect of a negative $\corrsr[]{}$ is to raise significantly the upper bound limit of source and relay power consumption. With fixed $\corrsr[]{}=-1$ it is enough to control the power of the nodes in order to meet: $\pathgain[\sourceindex]{\destindex}\power[\sourceindex]{}=\pathgain[\relayindex]{\destindex}\power[\relayindex]{}$. This equation is resolved by changing $\pathgain[\sourceindex]{\relayindex}$ to $\pathgain[\sourceindex]{\destindex}$ in \equationname~\ref{eq:fardest}.

\section{Amplify and forward technique}\label{sec:AFTech111}

In the AF technique, the transmit message $\signal[\sourceindex]{}$ is a sequence of $\blocknumber$ sub{-}messages $\signal[\sourceindex]{1},...,\signal[\sourceindex]{\blockindex},...,\signal[\sourceindex]{\blocknumber}$ which are independently and uniformly drawn from the message set $\signalset[\sourceindex]=\{0, 1, ..., 2^{\signalnumber\rate[\sourceindex]{}}-1\}$. Each sub-message $\signal[\sourceindex]{\blockindex}$ is separately encoded to $\var[\sourceindex][\blockindex]\left(\signal[\sourceindex]{\blockindex}\right)$ under the constraint that $\expectation[{\var[\sourceindex]^2}]\le\powermax[\sourceindex]{}$. In each block index $\blockindex$, the source node broadcasts $\var[\sourceindex][\blockindex]\left(\signal[\sourceindex]{\blockindex}\right)$, and at the same time the relay just increase the amplitude of the analog observed signal $\othervar[\relayindex][\blockindex-1]$ to result a normalized transmit $\var[\relayindex][\blockindex]\left(\signal[\sourceindex]{\blockindex-1}\right)$ as:
\begin{equation}\label{eq:AFrelay}
\begin{split}
\var[\relayindex][\blockindex]\left(\signal[\sourceindex]{\blockindex-1}\right)&=\AFconstant[][\blockindex].\othervar[\relayindex][\blockindex-1]\\
&=\AFconstant[][\blockindex].\left(\sqrt{\pathgain[\sourceindex]{\relayindex}}\,\var[\sourceindex][\blockindex-1]\left(\signal[\sourceindex]{\blockindex-1}\right)+\noise[\relayindex]\right)
\end{split}
\end{equation}
wherein $\AFconstant[]$ is amplification factor and is chosen to satisfy the relay's power limit. The relay node has its own power constraint as $\expectation[{\var[\relayindex]^2}]\le\powermax[\relayindex]{}$, so that:
\begin{align}\label{eq:AFfactorlimit}
|\AFconstant[][\blockindex]|^{2}\le\frac{\powermax[\relayindex]{}}{\,\noisePower + \pathgain[\sourceindex]{\relayindex}\powermax[\sourceindex]{}}
\end{align}

For simplicity, we assume $\,\AFconstant[][\blockindex]=\AFconstant[]$ in every block. As can be observed, if $\,\noisePower + \pathgain[\sourceindex]{\relayindex}\powermax[\sourceindex]{}\gg\powermax[\relayindex]{}\,$ the effect of the relay is negligible. Combining \equationname~\ref{eq:destreceive} and \equationname~\ref{eq:AFrelay} gives:
\begin{multline}\label{eq:AFdestreceive}
\othervar[\destindex][\blockindex]= \sqrt{\pathgain[\sourceindex]{\destindex}}\var[\sourceindex][\blockindex]\left(\signal[\sourceindex]{\blockindex}\right) + |\AFconstant[]|.\sqrt{\pathgain[\sourceindex]{\relayindex}\pathgain[\relayindex]{\destindex}}\,\var[\sourceindex][\blockindex-1]\left(\signal[\sourceindex]{\blockindex-1}\right)
\\
 + |\AFconstant[]|.\sqrt{\pathgain[\relayindex]{\destindex}}\noise[\relayindex] + \noise[\destindex]
\end{multline}
%%By \equationname~\ref{eq:AFdestreceive},
Thus, the maximum capacity of AF scheme turns out to be:
\begin{equation}\label{eq:AFcapacity111}
\displaystyle\upperbound[AF]{(\sourceindex;\relayindex;\destindex)}\!=\! \capacitySymbol\left(\frac{\left(\sqrt{\pathgain[\sourceindex]{\destindex}}\,+\,|\AFconstant[]|.\,\sqrt{\pathgain[\sourceindex]{\relayindex}\pathgain[\relayindex]{\destindex}}\,\right)^2\,.\,\powermax[\sourceindex]{}}{\left(1+|\AFconstant[]|^{2}\pathgain[\relayindex]{\destindex}\right)\noisePower}\right)
\end{equation}

The relay node does not regenerate any new code, and consequently the complexity of this scheme is low. Since the relay node amplifies whatever it receives, including noise, it is mainly useful in high SNR environments.
When the channel between the transmitter and the relay is very noisy, increasing the amplification factor $\AFconstant[]$ increases the noise at the destination. The relay should thus not always transmit with maximum power.
%%Increasing the amplification factor $\AFconstant[]$ increases the noise of ISI at the destination, and also increase the noise of ICI where there is only one antenna at the destination. The relay should thus not always transmit with maximum power.
Reference \cite[p. 46]{LiuBook09} demonstrates that under the condition:
%%\begin{equation}\label{eq:AFcoefMRC111}
 $|\AFconstant[]| \le \SNR[\sourceindex]{\relayindex}$,
%%\end{equation}
the $\displaystyle\upperbound[AF]{(\sourceindex;\relayindex;\destindex)}$ outperforms the capacity of the maximal ratio combining (MRC) technique that is:
\begin{equation}\label{eq:MRC111}
\displaystyle\upperbound[MRC]{(\sourceindex;\relayindex;\destindex)}=\displaystyle\capacitySymbol\left(\SNR[\sourceindex]{\destindex}+\frac{\SNR[\sourceindex]{\relayindex}\,\SNR[\relayindex]{\destindex}}{\SNR[\sourceindex]{\relayindex}+\SNR[\relayindex]{\destindex}}\right)
\end{equation}
 %%with $\displaystyle\SNR[\sourceindex]{\destindex}=\frac{\powermax[\sourceindex]{}\pathgain[\sourceindex]{\destindex}}{\noisePower}$,  $\displaystyle\SNR[\sourceindex]{\relayindex}=\frac{\powermax[\sourceindex]{}\pathgain[\sourceindex]{\relayindex}}{\noisePower}$ and $\displaystyle\SNR[\relayindex]{\destindex}=\frac{\powermax[\relayindex]{}\pathgain[\relayindex]{\destindex}}{\noisePower}$.

Comparing \eqref{eq:AFcapacity111} and \eqref{eq:MRC111} we find that the MRC technique performs better than AF under the following condition:
\begin{equation}
\displaystyle\frac{\SNR[\sourceindex]{\relayindex}\,\SNR[\sourceindex]{\destindex}}{\SNR[\sourceindex]{\relayindex}+\SNR[\sourceindex]{\destindex}} < |\AFconstant[]|^{2}.\pathgain[\relayindex]{\destindex}
\end{equation}

\section{Case study}\label{sec:casestudy111}

\begin{figure}
  \begin{center}
    \psfrag{s}[c][c][0.9]{source ($\sourceindex$)}
    \psfrag{r}[c][c][0.9]{relay ($\relayindex$)}
    \psfrag{d}[c][c][0.9]{dest. ($\destindex$)}
    \psfrag{d1}[c][c][1]{$d_{\sourceindex\relayindex}$}
    \psfrag{h}[c][c][1]{$d_{\relayindex}$}
    \psfrag{1}[c][b][1]{$d_{\sourceindex\destindex}$}
    \includegraphics[width=0.5\columnwidth]{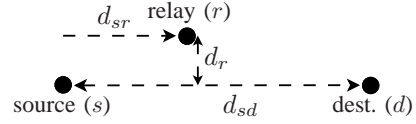}
  \end{center}
  \caption{A single relay communication network scenario.}
  \label{fig:Net_R111}
\end{figure}
In this section, we exemplify the various outer region bounds presented so far in this section. We consider a point{-}to{-}point relayed communication with Gaussian channels wherein the transmitter $\sourceindex$, relay $\relayindex$, and sink $\destindex$ are located as sketched in \figurename~\ref{fig:Net_R111}. We assume a vertical distance of $d_{\relayindex}$ between the relay and the $\sourceindex\to\destindex$ direct{-}link. The path condition values ${\pathgain[\sourceindex]{\relayindex}=\pathgain[\relayindex]{\destindex}=\pathgain[\sourceindex]{\destindex}=1}$ are scaled with respect to Notation~\ref{nota:scaledpathgain}.
First, we experiment a high SNR environment and suppose the source and destination are located at a distance of $d_{\sourceindex\destindex}=1\m$, and the relay is located at a vertical distance of $d_{\relayindex}=0.1\m$ and it is horizontally moving from $d_{\sourceindex\relayindex}=-0.5\m$ to $d_{\sourceindex\relayindex}=1.5\m$. \figurename~\ref{fig:high_111} plots various data rates for $\powermax[\sourceindex]{}=\powermax[\relayindex]{}=100\mW$, and $\noisePower=1\muW$. The curve labeled AF shows the outer region of AF strategy with the largest possible scaling factor $\AFconstant[]$ in \equationname~\ref{eq:AFfactorlimit}.
The curve labeled $\corrsr[]{}$ plots a particular value of the correlation coefficient we tried for this example.
%%It is important to note that, changing the correlation function meaningfully results in change of the curves labeled cutset, DF, and CF (strategies).

\begin{figure}
  \begin{center}
    \psfrag{xaxis}[c][c][0.9]{$d_{\sourceindex\relayindex} \left[\m\right]$}
    \psfrag{yaxis}[c][c][0.8]{Rate $\left[\bpsHz\right]$}
    \includegraphics[width=0.75\columnwidth]{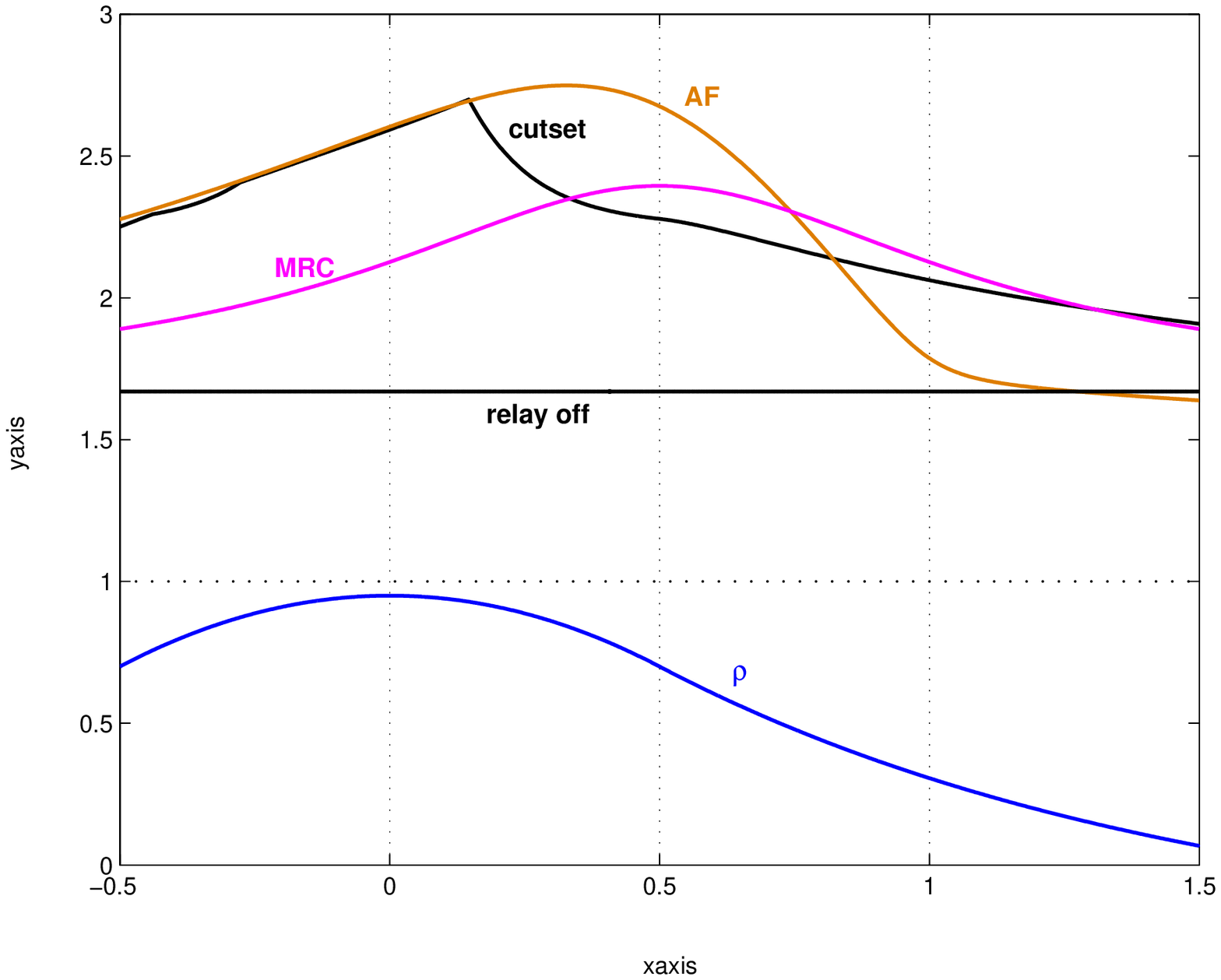}
  \end{center}
  \caption{Rates for one relay with $\powermax[\sourceindex]{}=\powermax[\relayindex]{}=100\mW$, $\noisePower=1\muW$,  $d_{\sourceindex\destindex}=1\m$, and $d_{\relayindex}=0.1\m$.}
  \label{fig:high_111}
\end{figure}

In the studied case study, the AF and MRC techniques show a very good performance.
This is due to the fact of the short distances
%, and high power bounds
result in a high SNR.  As the relay moves toward the destination ($d_{\sourceindex\relayindex}\to 1$), the achieved signal at the relay becomes weaker and this significantly decreases the AF data rate. Generally, the AF and MRC techniques would be useful when the relay is located so as to be able to perfectly receive and deliver signals, or equivalently, the relay is equidistant from the transmitter and the sink ($d_{\sourceindex\relayindex}\to 0.5$). %%In a low SNR environment, DF and CF coding techniques show better performance.

As the relay moves toward transmitter ($d_{\sourceindex\relayindex}\to 0$), the cutset data rate is equal to the second term of \equationname~\ref{eq:AWGNcutset}.
%%%\begin{equation*}
%%%\upperbound[cutset]{(\sourceindex;\relayindex;\destindex)}=
%%%\capacitySymbol\left(\frac{\pathgain[\sourceindex]{\destindex}\powermax[\sourceindex]{}+\pathgain[\relayindex]{\destindex}\powermax[\relayindex]{}+ 2\corrsr[]{}\sqrt{\pathgain[\sourceindex]{\destindex}\powermax[\sourceindex]{}\pathgain[\relayindex]{\destindex}\powermax[\relayindex]{}}}{\noisePower}\right)\\
%%%\end{equation*}
Correspondingly, that is equal to the first term of \equationname~\ref{eq:AWGNcutset}, as the relay is placed close to the destination ($d_{\sourceindex\relayindex}\to 1$).
%%%\begin{equation*}
%%%\upperbound[cutset]{(\sourceindex;\relayindex;\destindex)}=
%%%\capacitySymbol\left(\frac{\left(\pathgain[\sourceindex]{\destindex}\!+\! \pathgain[\sourceindex]{\relayindex}\right)\powermax[\sourceindex]{}\left(1\!-\!\corrsr[]{}^2\right)}{\noisePower}\right)
%%%\end{equation*}

Now, we consider a point{-}to{-}point relayed connection in a low SNR regime. The transmitter and the destination are placed at a distance of $d_{\sourceindex\destindex}=500\m$, and the relay is moving in a range of $d_{\sourceindex\relayindex}=-100\div600\m$ with a vertical distance of $d_{\relayindex}=10\m$.
%%For the sake of clarity of the figure, we suppose every channels have a bandwidth of ${\carrierSpacing=70\kHz}$. Thus, the data rate is calculated as: ${\capacitySymbol(\SNR[]{}) = \frac{\carrierSpacing}{2}\log_{2}{(1+\SNR[]{})}}$.
We set the same $\AFconstant[], \powermax[\sourceindex]{}, \powermax[\relayindex]{}$, and $\noisePower$ as the previous simulation.
\figurename~\ref{fig:low_111} plots various data rates.
%%for $\powermax[\sourceindex]{}=\powermax[\relayindex]{}=100\mW$, and $\noisePower=1\muW$.
\begin{figure}
  \begin{center}
    \psfrag{xaxis}[c][c][0.9]{$d_{\sourceindex\relayindex} \left[\m\right]$}
    \psfrag{y1axis}[c][c][0.8]{Rate $\left[\bpsHz\right]$}
    \psfrag{y2axis}[c][c][0.85]{\color{blue}{Correlation coefficient $\corrsr[]{}$}}
    \includegraphics[width=0.8\columnwidth]{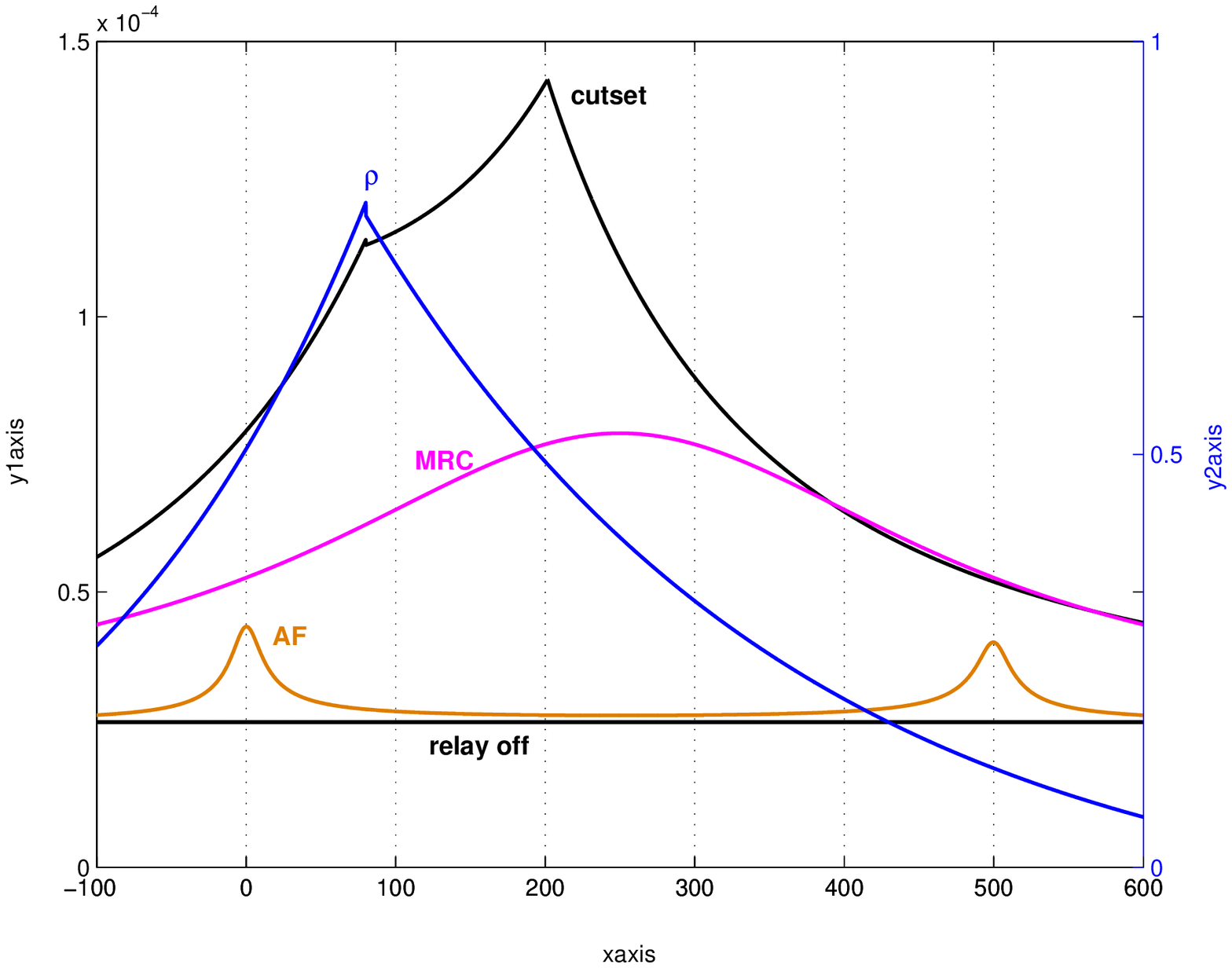}
  \end{center}
  \caption{Rates for one relay with $\powermax[\sourceindex]{}=\powermax[\relayindex]{}=100\mW$, $\noisePower=1\muW$, ${d_{\sourceindex\destindex}=500\m}$, and $d_{\relayindex}=10\m$.}
  \label{fig:low_111}
\end{figure}

We draw a different experimental function for correlation value which is the curve labeled $\corrsr[]{}$.
From \figurename~\ref{fig:low_111}, it is clearly derived that the AF technique is not quite useful in a low SNR network, whereas the MRC technique performs much better than AF.
This is because the received signal at the relay is very noisy, and also the scaling factor is higher than that in the previous scenario. In this situation (almost) no signal perfectly approaches to the destination.
%%The unique point that hold \equationname~\ref{eq:AFcoefMRC111} is at $d_{\sourceindex\relayindex}=0$.
Our experiments in a given scenario with different parameters result that reducing the amplification factor does not effect the AF data rate.
As can be seen, in low SNR, the coding technique shows significantly higher rates than in the direct{-}link.

\section{Summary}\label{sec:relay111dis}

The data rate of a relayed communication is a function of both relaying strategy and positions of the nodes.
The cutset upper bound capacity is decomposed into two terms: broadcast capacity from transmitter's viewpoint and multiple channel access at the destination.
When the relay is close to the transmitter, the upper bound capacity of relaying communication is equal to that multiple{-}access channel. On the other hand, when the relay is close to the destination, the upper bound capacity of relaying communication is equal to that broadcast channel.
In a low SNR environment, cutset technique achieves a significantly high data rate, whereas AF shows a good performance in high SNR regime.
%%It has been depicted in \cite[Fig. 16]{Kramer05} that DF and CF protocols can be capacity achieving depending on the location of the relay node.

\setlength{\IEEEilabelindent}{2\IEEEiedmathlabelsep}
\IEEEusemathlabelsep

\bstctlcite{mybibfile:BSTcontrol}
\bibliography{IEEEabrv,mybibfile}

\end{document}